\documentclass[fleqn,10pt]{wlscirep}

\usepackage{epsfig}
\usepackage{amsmath}
\usepackage{graphicx}
\usepackage{dcolumn}
\usepackage{bm}
\usepackage{epstopdf}
\usepackage{bbold}
\usepackage{amsthm}
\usepackage{subfigure}
\newtheorem{theorem}{Theorem}

\title{Extracting quantum coherence via steering}

\author[1,*]{Xueyuan Hu}
\author[2]{Heng Fan}
\affil[1]{School of Information Science and Engineering, and Shandong Provincial Key Laboratory of Laser Technology and Application, Shandong University, Jinan, 250100, P. R. China}
\affil[2]{Institute of Physics, Chinese Academy of Sciences, Beijing 100190, China}

\affil[*]{xyhu@sdu.edu.cn}

\begin{abstract}
As the precious resource for quantum information processing, quantum coherence
can be created remotely if the involved two sites are quantum correlated.
It can be expected that the amount of coherence created should depend on
the quantity of the shared quantum correlation, which is also a resource.
Here, we establish an operational connection between coherence induced by steering
and the quantum correlation. We find that the steering-induced coherence quantified by such as
relative entropy of coherence and trace-norm of coherence
is bounded from above by a known quantum correlation measure defined as the one-side measurement-induced disturbance.
The condition that the upper bound
saturated by the induced coherence varies for different measures of coherence.
The tripartite scenario is also studied and similar conclusion can be obtained.
Our results provide the operational connections between local and non-local resources in quantum information processing.
\end{abstract}
\begin{document}

\flushbottom
\maketitle
%
%
\thispagestyle{empty}



Quantum coherence, being at the heart of quantum mechanics, plays a key role in quantum information processing
such as quantum algorithms \cite{INSPEC:4865020} and quantum key distribution \cite{PhysRevLett.68.3121}.
Inspired by the recently proposed resource theory of quantum coherence \cite{PhysRevLett.113.140401,PhysRevLett.115.070503}, researches are focused on the quantification \cite{PhysRevA.92.022124,PhysRevLett.113.170401} and evolution \cite{PhysRevLett.115.210403,PhysRevLett.115.113002} of quantum coherence, as well as its operational meaning \cite{PhysRevA.92.022124,arXiv:1506.07975} and role in quantum information tasks \cite{PhysRevLett.113.030502,arXiv:1505.03792,arXiv:1509.01300}.
When multipartite systems are considered, coherence is closely related to the well-established quantum information resources, such as entanglement \cite{RevModPhys.81.865} and discord-type quantum correlations \cite{RevModPhys.84.1655}. It is shown that the coherence of an open system is frozen under the identical dynamical condition where discord-type quantum correlation is shown to freeze \cite{PhysRevLett.114.210401}. Further, discord-type quantum correlation can be interpreted as the minimum coherence of a multipartite system on tensor-product basis \cite{arXiv1506.01773}. An operational connection between local coherence and non-local quantum resources (including entanglement \cite{PhysRevLett.115.020403} and discord \cite{arXiv:1510.06179}) is presented. It is shown that entanglement or discord between a coherent system and an incoherent ancilla can be built
by using incoherent operations, and the generated entanglement or discord is bounded from above by the initial coherence. The converse procedure is of equal importance: to extract coherence locally from a spatially separated but quantum correlated bipartite state. The extraction of coherence with the assistance of a remote party has been studied in the asymptotical limit \cite{arXiv:1507.08171}. In this paper, we ask how we extract coherence locally from a \emph{single copy} of a bipartite state.

The quantum steering has long been noted as a distinct nonlocal quantum effect \cite{steer35} and has attracted recent research interest both theoretically and experimentally \cite{PhysRevLett.98.140402,Nat.Phys.2010,Nat.Photon.2012,PhysRevLett.112.180404,Verstraete2002,Shi2011,PhysRevLett.113.020402,steer_njp,PhysRevA.92.012311,PhysRevA.91.022301,arXiv:1507.02358}. It demonstrates that Alice can remotely change Bob's state by her local selective measurement if they are correlated, and is hence a natural candidate to accomplish the task of remote coherence extraction.

In this paper, we present the study of coherence extraction induced by quantum steering and the involved quantum correlation. Precisely, we introduce the quantity of steering-induced coherence (SIC) for bipartite quantum states. 
Here Bob is initially in an incoherent state but quantum correlated to Alice. Alice's local projective measurement can thus steer Bob to a new state which might be coherent. The SIC $\bar{\mathcal C}$ is then defined as the maximal average coherent of Bob's steered states that can be created by Alice's selective projective measurement. When there is no obvious incoherent basis for Bob, (for example, Bob's system is a polarized photon), the definition can be generalized to arbitrary bipartite system where Bob's incoherent basis is chosen as the eigenbasis of his reduced state. In this case, the SIC can be considered as a basis-free measure of Bob's coherence. The main result of this paper is building an operational connection between the SIC and the shared quantum correlation between Alice and Bob. We prove that the SIC can not surpass the initially shared $B$-side quantum correlation, which is a known quantum correlation measure named as measurement-induced disturbance (MID) $\mathcal Q_B$ \cite{PhysRevA.77.022301}. States whose relative entropy SIC $\bar{\mathcal C}^r$ can reach its upper bound $\mathcal Q_B^r$ are identified as maximally correlated states. For two-qubit states, while the trace-norm SIC $\bar{\mathcal C}^{t}$ can always reach the corresponding $\mathcal Q_B^t$, we find an example of two-qubit state whose $\bar{\mathcal C}^r$ is strictly less than $\mathcal Q_B^r$. This indicates that the condition for $\bar{\mathcal C}$ to reach the upper bound strongly depends on the measure of coherence. We further generalize the results to a tripartite scenario, where Alice can induce entanglement between Bob and Charlie in a controlled way. Since coherence of a single party is generally robust than quantum correlations involving two parties, our work provides a way to ``store'' quantum correlation as coherence. Besides, the coherent state induced by steering
can be widely used for quantum information processing. Our results establish the intrinsic connection between coherence and quantum correlation by steering.

\section*{Results}

\subsection*{Coherence and measurement-induced disturbance}
A state is said to be incoherent on the reference basis $\Xi=\{|\xi_i\rangle\}$, if it can be written as \cite{PhysRevLett.113.140401}
\begin{equation}
\sigma_\Xi=\sum_ip_i|\xi_i\rangle\langle \xi_i|.
\end{equation}
Let $\mathcal I_\Xi$ be the set of incoherent state on basis $\Xi$. The incoherent completely positive trace-preserving (ICPTP) channel is defined as
\begin{equation}
\Lambda_{\mathrm{ICPTP}(\cdot)}=\sum_nK_n(\cdot)K_n^\dagger,
\end{equation}
where the Kraus operators $K_n$ satisfy $K_n\mathcal I_\Xi K_n^\dagger\subset\mathcal I_\Xi$. According to Ref. \cite{PhysRevLett.113.140401}, a proper coherence measure $C(\rho,\Xi)$ of a quantum state $\rho$ on a fixed reference basis $\Xi$ should satisfy the following three conditions. (C1) $C(\rho,\Xi)=0$ iff $\rho\in\mathcal I_\Xi$. (C2) Monotonicity under selective measurements on average: $C(\rho,\Xi)\geq\sum_np_nC(\rho_n,\Xi),\ \forall \{K_n\}$ satisfying $K_n\mathcal I_\Xi K_n^\dagger\subset\mathcal I_\Xi$ and $\sum_nK_n^\dagger K_n=I$, where $\rho_n=K_n\rho K_n^\dagger/p_n$, occurring with probability $p_n=\mathrm{tr}[K_n\rho K_n^\dagger]$, is the state corresponding to outcome $n$. (C3) Convexity: $\sum_np_nC(\rho_n,\Xi)\geq C(\sum_np_n\rho_n,\Xi)$.

A candidate of coherence measure is the minimum distance between $\rho$ and the set of incoherent states
\begin{equation}
C(\rho,\Xi)=\min_{\sigma\in\mathcal I_\Xi} D(\rho,\sigma),\label{coherence}
\end{equation}
where $D(\cdot,\cdot)$ is a distance measure on quantum states and satisfies the following five conditions. (D1) $D(\rho,\sigma)=0$ iff $\rho=\sigma$. (D2)  Monotonicity under selective measurements on average: $D(\rho,\sigma)\geq\sum_np_nD(\rho_n,\sigma),\ \forall \{K_n\}$. (D3) Convexity: $\sum_np_nD(\rho_n,\sigma)\geq D(\sum_np_n\rho_n,\sigma)$. (D4) $D(\rho,\Lambda^\Xi(U\rho U^\dagger))\geq D(\rho,\Lambda^\Xi(\rho))$, $\forall\Xi,U$, where $U$ is a unitary operation, and $\Lambda^\Xi$ denotes the projective measurement on basis $\Xi$: $\Lambda^\Xi(\cdot)\equiv\sum_i|\xi_i\rangle\langle\xi_i|(\cdot)|\xi_i\rangle\langle\xi_i|$. (D5) $D(\rho,\sigma)=D(\rho\otimes\rho_a,\sigma\otimes\rho_a)$. Conditions (D1-D3) make sure that (C1-C3) is satisfied by the coherence measure defined in Eq. (\ref{coherence}). When (D4) is satisfied, the coherence of $\rho$ on the reference basis $\Xi$ can be written as
\begin{equation}
C(\rho,\Xi)=D(\rho,\Lambda^\Xi(\rho)).\label{coherence1}
\end{equation}
As proved in Ref. \cite{PhysRevLett.113.140401}, the relative entropy $D^r(\rho,\sigma)=S(\rho||\sigma)\equiv\mathrm{Tr}(\rho\log_2\rho-\rho\log_2\sigma)$ and the $l_1$ matrix norm $D^{l_1}(\rho,\sigma)=\|\rho-\sigma\|_{l_1}\equiv\sum_{ij}|\rho_{ij}-\sigma_{ij}|$ satisfies all the conditions (D1-D4), which makes the corresponding coherence measures $C^r(\rho,\Xi)=D^r(\rho,\Lambda^\Xi(\rho))$ and $C^{l_1}(\rho,\Xi)=D^{l_1}(\rho,\Lambda^\Xi(\rho))$ satisfy the conditions (C1-C3). As discovered recently \cite{arXiv:1606.03181}, the trace-norm distance $D^t(\rho,\sigma)\equiv\mathrm{tr}\sqrt{(\rho_B-\sigma_B)^\dagger(\rho_B-\sigma_B)}$ does not satisfy (D2).

Introduced in Ref. \cite{PhysRevA.77.022301}, MID characterizes the quantumness of correlations. MID of a bipartite system $\rho$ is defined as the minimum disturbance caused by local projective measurements that do not change the reduced states $\rho_A\equiv\mathrm{Tr}_B(\rho)$ and $\rho_B\equiv\mathrm{Tr}_A(\rho)$
\begin{equation}
\mathcal Q(\rho)=\inf_{\mathbb E_A,\mathbb E_B}D(\rho,\Lambda_A^{\mathbb E_A}\otimes\Lambda_B^{\mathbb E_B}(\rho)),\label{MID}
\end{equation}
where the infimum is taken over projective measurements which satisfy $\Lambda_A^{\mathbb E_A}(\rho_A)=\rho_A$ and $\Lambda_B^{\mathbb E_B}(\rho_B)=\rho_B$, and $D(\cdot,\cdot)$ is a distance on quantum states, which satisfies conditions (D1-D5) and further (D6) $D(U\rho U^\dagger,U\sigma U^\dagger)=D(\rho,\sigma)$. It can be checked that (D6) can be satisfied by relative entropy but not satisfied by $l_1$-norm. Comparing Eq. (\ref{MID}) with Eq. (\ref{coherence1}), we find MID is just the coherence of the bipartite state $\rho$ on the local eigenbasis $\mathbb E_A\otimes\mathbb E_B$.

For later convenience, we introduce $B$-side MID as
\begin{equation}
\mathcal Q_B(\rho)=\inf_{\mathbb E_{B}:\Lambda_{B}^{\mathbb E_B}(\rho_{B})=\rho_B}D(\rho,I_A\otimes\Lambda_B^{\mathbb E_B}(\rho)).
\end{equation}
$\mathcal Q_B$ goes to zero for $B$-side classical states, which can be written as $\rho_{B-cla}=\sum_i\rho_i^A\otimes|e_i^B\rangle\langle e_i^B|$, while $\mathcal Q$ is strictly positive for $\rho_{B-cla}$ if $\exists i,\ [\rho_A,\rho_i^A]\neq0$. Notice that for $\mathcal Q_B$ one do not have a coherence interpretation.

\subsection*{Definition of steering-induced coherence}
As shown in Fig. \ref{SIC}, Alice and Bob initially share a quantum correlated state $\rho$, and Bob's reduced state $\rho_B$ is incoherent on his own basis. Now Alice implements a local projective measurement on basis $\Xi_A$. When she obtains the result $i$ (which happens with probability $p^{\xi_i}$), Bob is ``steered'' to a coherent state $\rho_B^{\xi_i}$. We introduce the concept of SIC for characterizing Alice's ability to create Bob's coherence on average using her local selective measurement.

\begin{figure}
\includegraphics[width=0.8\columnwidth]{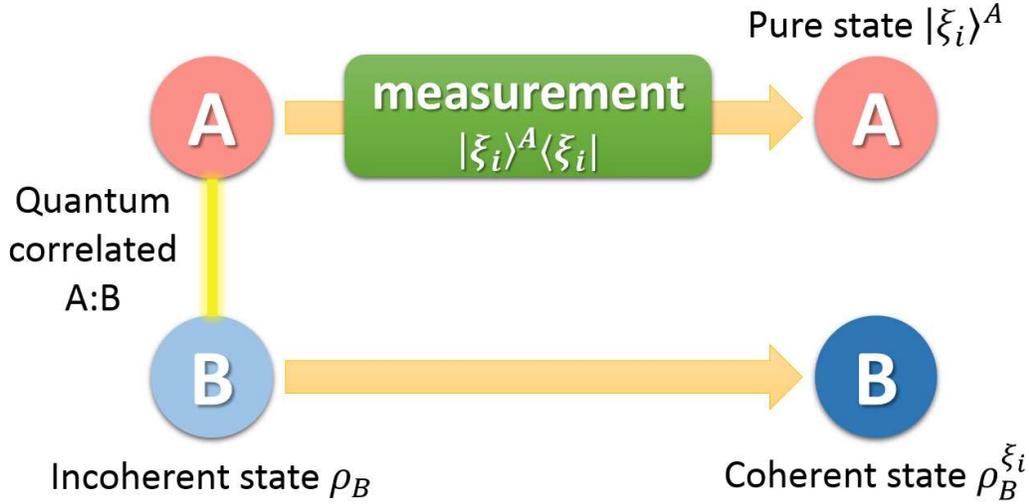}
\caption{\label{SIC}(color online). Scheme for creating Bob's coherence by Alice's local measurement and classical communication. When Alice implements local projective measurement on basis $\Xi_A=\{|\xi_i^A\rangle\}$, she gets result $i$ with probability $p^{\xi_i}$ and meanwhile steer Bob's state to $\rho^{\xi_i}_B$ which can be coherent on Bob's initial eigenstate $\mathbb E_B$. SIC is defined as the maximal average coherence of states $\rho^{\xi_i}_B$ that can be created by Alice's local selective measurement.}
\end{figure}

\textbf{Definition (Steering-induced coherence, SIC)}. \emph{For a bipartite quantum state $\rho$, Alice implements projective measurement on basis $\Xi_A=\{\xi_i^A\}$ ($i=0,\cdots,d_A-1)$. With probability $p^{\xi_i}=\mathrm{tr}[\rho(\xi^A_i\otimes I)]$, she obtains the result $\xi^A_i=|\xi^A_i\rangle\langle\xi^A_i|$, which steers Bob's state to $\rho^{\xi_i}_B=\langle\xi_i^A|\rho|\xi^A_i\rangle/p^{\xi_i}$. Let $\mathbb E_B=\{|e_j^B\rangle\}$ ($j=0,\cdots,d_B-1$) be the eigenbasis of reduced states $\rho_B$. The steering-induced coherence is defined as the maximum average coherence of Bob's steered states on the reference basis $\mathbb E_B$
\begin{equation}
\bar{\mathcal C}(\rho)=\inf_{\mathbb E_B}\left[\max_{\Xi_A}\sum_ip^{\xi_i}C(\rho_B^{\xi_i},\mathbb E_B)\right].\label{SCoh}
\end{equation}
where the maximization is taken over all of Alice's projective measurement basis $\Xi_A$, and the infimum over $\mathbb E_{B}$ is taken when $\rho_{B}$ is degenerate and hence $\mathbb E_{B}$ is not unique.}

Since Bob's initial state $\rho_B$ is incoherent on its own basis $\mathbb E_B$, the SIC $\bar{\mathcal C}(\rho)$ describes the maximum ability of Alice's local selective measurement to create Bob's coherence on average. We verify the following properties for $\bar{\mathcal C}(\rho)$.\\
(E1) $\bar{\mathcal C}(\rho)\geq0$, and $\bar{\mathcal C}(\rho)=0$ iff $\rho$ is a $B$-side classical state.\\
(E2) Non-increasing under Alice's local completely-positive trace-preserving channel: $\bar{\mathcal C}(\Lambda_A\otimes I(\rho))\leq\bar{\mathcal C}(\rho)$.\\
(E3) Monotonicity under Bob's local selective measurements on average: $\bar{\mathcal C}(\rho)\geq\sum_np_n\bar{\mathcal C}(\rho_n),\ \forall \{K_n^B\}$ satisfying $K_n^B\mathcal I_{\mathbb E_B}K_n^{B\dagger}\subset\mathcal I_{\mathbb E_B}$, where $\rho_n=I_A\otimes K_n^B\rho(I_A\otimes K_n^B)^\dagger/p_n$ and $p_n=\mathrm{tr}[I_A\otimes K_n^B\rho(I_A\otimes K_n^B)^\dagger]$.\\
(E4) Convexity: $\sum_np_n\bar{\mathcal C}(\rho_n)\geq\bar{\mathcal C}(\sum_np_n\rho_n)$.

\begin{proof}
Condition (E1) can be proved using the method in Ref. \cite{arXiv:1507.02358}, where it is proved that $\mathcal C_i(\rho)\equiv\max_{\xi^A_i}C(\rho_B^{\xi_i},\mathbb E_B)$ vanishes iff $\rho$ is a $B$-side classical state. (E2) is verified by noticing that the local channel $\Lambda_A$ can not increase the set of Bob's steered states, and hence the optimal steered states $\{\rho_B^{\xi_i}\}$ may not be steered to after the action of channel $\Lambda_A$. The conditions (E3) and (E4) are directly derived from conditions (C2) and (C3) for coherence.
\end{proof}

\subsection*{Relation between SIC and MID}
Intuitively, Alice's ability to extract coherence on Bob's side should depend on the quantum correlation between them. The following theorem gives a quantitative relation between the SIC $\bar{\mathcal C}(\rho)$ and quantum correlation measured by $B$-side MID $\mathcal Q_B(\rho)$.

\begin{theorem}\label{thm:ubound}
When the distance measure in the definition of MID and coherence satisfies conditions (D1-D6), the SIC is bounded from above by the $B$-side MID, i.e.,
\begin{equation}
\bar{\mathcal C}(\rho)\leq \mathcal Q_B(\rho).\label{th1}
\end{equation}

\begin{proof}We start with the situation that $\rho_B$ is non-degenerate and hence one do not need to take the infimum in Eqs. (\ref{MID}) and (\ref{SCoh}). By definition, we have
\begin{equation}
\mathcal Q_B(\rho)=D(\rho,\rho^{\mathbb E_B}),
\end{equation}
where $\rho^{\mathbb E_B}=I\otimes\Lambda^{\mathbb E_B}(\rho)$.

After Alice implements a selective measurement on basis $\Xi_A$, the average coherence of Bob's state becomes
\begin{eqnarray}
\bar{\mathcal C}_{\Xi_A}(\rho)&=&\sum_ip^{\xi_i}D(\rho_B^{\xi_i},\Lambda^{\mathbb E_B}(\rho_B^{\xi_i}))\nonumber\\
&=&\sum_ip^{\xi_i}D\left(\frac{\xi_i^A\rho\xi_i^{A\dagger}}{p^{\xi_i}},\frac{\xi_i^A\rho^{\mathbb E_B}\xi_i^{A\dagger}}{p^{\xi_i}}\right).
\end{eqnarray}
The second equality holds because $D(\rho_B^{\xi_i},\Lambda^{\mathbb E_B}(\rho_B^{\xi_i}))=D(\xi_i^A\otimes\rho_B^{\xi_i},I_A\otimes\Lambda^{\mathbb E_B}(\xi_i^A\otimes\rho_B^{\xi_i}))$ (condition (D5)) and $\xi_i^A\otimes\rho_B^{\xi_i}=\frac{\xi_i^A\rho\xi_i^{A\dagger}}{p^{\xi_i}}$. Since selective measurement does not increase the state distance (condition (D2)), we have $\bar{\mathcal C}_{\Xi_A}(\rho)\leq\mathcal Q_B(\rho),\forall\Xi_A$, and hence Eq. (\ref{th1}) holds.

The generalization to degenerate state is straightforward. We choose $\mathbb E^o_B$ to reach the infimum of $\mathcal Q_B$, which may not be the optimal eigen-basis for $\bar{\mathcal C}$. Hence we have $\mathcal Q_B(\rho)=D\left(\rho,I_A\otimes \Lambda^{\mathbb E^o_B}_B(\rho)\right)\geq\max_{\Xi_A}\sum_ip^{\xi_i}C(\rho_B^{\xi_i},\mathbb E^o_B)\geq\bar{\mathcal C}(\rho)$.
\end{proof}
\end{theorem}

\begin{figure}
\includegraphics[width=0.8\columnwidth]{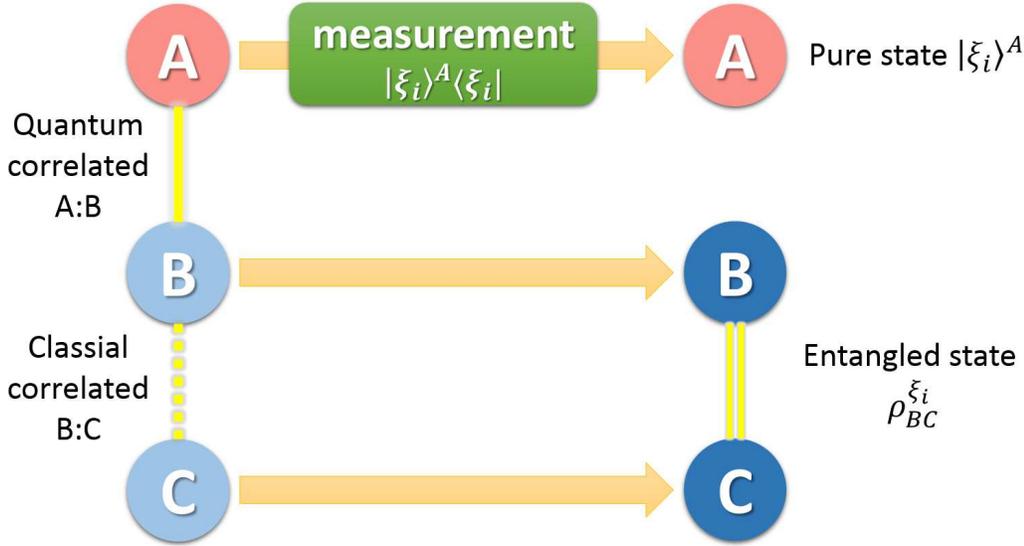}
\caption{\label{SIE}(color online). Scheme for creating entanglement between Bob and Charlie by Alice's local selective measurement. When Alice implements local projective measurement on basis $\Xi_A=\{|\xi_i^A\rangle\}$, she gets result $i$ with probability $p^{\xi_i}$ and meanwhile steer the state shared between Bob and Charlie to $\rho^{\xi_i}_{BC}$ which can be entangled.}
\end{figure}

According to Ref. \cite{PhysRevLett.115.020403}, the coherence of a quantum system $B$ can in turn be transferred to the entanglement between the system and an ancilla $C$ by incoherent operations. The established entanglement, measured by the minimum distance between the state $\rho^{BC}$ and a separable state $\sigma^{BC}:=\sum_kp_k\rho_k^B\otimes\rho_k^C$ as $E(\rho_{BC})=\min_{\sigma^{BC}\in\mathcal S}D(\rho^{BC},\sigma^{BC})$, is bounded from above by the initial coherence of $B$. Here $\mathcal S$ is the set of separable states and the state distance $D$ is required not to increase under trace-preserving channels $D(\Lambda(\rho),\Lambda(\sigma))\leq D(\rho,\sigma)$, which is automatically satisfied when we combine conditions (D2) and (D3).

This leads to the three-party protocol as shown in Fig. \ref{SIE}, where Alice's local selective measurement can create entanglement between Bob and Charlie. In this protocol, Bob and Charlie try to build entanglement between them from a product state $\varrho^{AB}\otimes|e_0^C\rangle\langle e_0^C|$, but are limited to use incoherent operations. Since $\varrho_B$ is incoherent on his eigenbasis $\mathbb E_B$, Bob and Charlie can build only classically correlated state without Alice's help. Now Alice implement projective measurement $\{|\xi_i^A\rangle\langle\xi_i^A|\}$ and on the outcome $i$, the state shared between Bob and Charlie is steered to $\rho^{\xi_i}_{BC}$ which can be entangled. The following corollary of theorem \ref{thm:ubound} gives the upper bound of the steering-induced entanglement.

\textbf{Corollary 1} \emph{Alice, Bob and Charlie share a tripartite state $\rho$, which is prepared from the product state $\varrho^{AB}\otimes|e_0^C\rangle\langle e_0^C|$ using an ICPTP channel on $BC$: $\rho=I_A\otimes\Lambda_{ICPTP}^{BC}(\varrho^{AB}\otimes|e_0^C\rangle\langle e_0^C|)$. Here $\mathbb E_B\otimes\mathbb E_C=\{|e_i^B\rangle\otimes|e_j^C\rangle\}$ is the reference basis of coherence. Alice's local selective measurement $\{|\xi_i^A\rangle\langle\xi_i^A|\}$ can establish entanglement between Bob and Charlie, and the established entanglement on average is bounded from above by the initial $B$-side MID between Alice and Bob
\begin{equation}
\sum_ip^{\xi_i}E(\rho^{\xi_i}_{BC})\leq\mathcal Q_B(\varrho^{AB}).\label{SIE_ubound}
\end{equation}}

\begin{proof}
Before Alice implement the measurement, the state shared between Bob and Charlie is incoherent on basis $\mathbb E_B\otimes\mathbb E_C$ and hence can be written as $\rho^{BC}=\sum_{ij}q_{ij}|e_i^B\rangle\langle e^B_i|\otimes|e_j^C\rangle\langle e^C_j|$. Apparently, $\mathcal Q(\rho^{BC})=0$, so Bob and Charlie is classically correlated.

On the measurement outcome $i$, the entanglement between Bob and Charlie becomes $E(\rho^{\xi_i}_{BC})$ which satisfies $E(\rho^{\xi_i}_{BC})\leq\mathcal Q(\rho^{\xi_i}_{BC})\leq C(\rho^{\xi_i}_{BC},\mathbb E_B\otimes\mathbb E_C)$. Notice that $\rho^{\xi_i}_{BC}=\Lambda_{ICPTP}^{BC}(\varrho^{\xi_i}_{B}\otimes|e_0^C\rangle\langle e_0^C|)$ and hence $C(\rho^{\xi_i}_{BC},\mathbb E_B\otimes\mathbb E_C)\leq C(\varrho^{\xi_i}_{B}\otimes|e_0^C\rangle\langle e_0^C|,\mathbb E_B\otimes\mathbb E_C)=C(\varrho^{\xi_i}_{B},\mathbb E_B)$. Eq. (\ref{SIE_ubound}) is arrived by noticing that $\sum_ip^{\xi_i}C(\varrho^{\xi_i}_{B},\mathbb E_B)\leq\mathcal Q_B(\varrho^{AB})$ from theorem 1.
\end{proof}

Now we consider a general tripartite state $\rho$. If the reduced state $\rho^{BC}=\mathrm{tr}_A{\rho}$ is non-degenerate, one can follow the same steps and prove that
\begin{equation}
\sum_ip^{\xi_i}E(\rho^{\xi_i}_{BC})\leq\mathcal Q_{\{BC\}}(\rho),\label{SIE_ubound1}
\end{equation}
whenever $\rho^{BC}$ is incoherent on basis $\mathbb E_B\otimes\mathbb E_C$. Here $\mathcal Q_{\{BC\}}$ is the $\{BC\}$-side MID between Alice and the combination of Bob and Charlie. However, when $\rho^{BC}$ is degenerate, the condition that the tripartite state $\rho$ is prepared from $\varrho^{AB}\otimes|e_0^C\rangle\langle e_0^C|$ by an ICPTP channel on $BC$ is stringent. For example, the state $\rho^X=\frac12|0\rangle^A\langle0|\otimes|\Psi^+\rangle^{BC}\langle\Psi^+|+\frac12|1\rangle^A\langle1|\otimes|\Psi^-\rangle^{BC}\langle\Psi^-|$ where $|\Psi^\pm\rangle=\frac{1}{\sqrt2}(|00\rangle\pm|11\rangle)$, with $\rho^{BC}$ incoherent on basis $\{|ij\rangle^{BC}\},(i,j=0,1)$, violates Eq. (\ref{SIE_ubound1}), since $\mathcal Q_{\{BC\}}(\rho)=0$ but the left-hand-side reaches unity for Alice's measurement basis $\Xi_A=\{|0\rangle,|1\rangle\}$. It indicate that the state $\rho^X$ can not be prepared from a product state in the form $\varrho^{AB}\otimes|e_0^C\rangle\langle e_0^C|$ using only incoherent operations.

\subsection*{States to reach the upper bound}
According to theorem \ref{thm:ubound}, Bob's maximal coherence that can be extracted by Alice's local selective measurement is bounded from above by the initial quantum correlation between them. Since the relative entropy is the only distance measure found to date which satisfies all the conditions (D1-D6), we employ relative entropy as the distance in the definition of coherence and MID, and discuss the states which can reach the upper bound of theorem \ref{thm:ubound}.

\begin{theorem}\label{thm:r} The SIC can reach $B$-side MID
\begin{equation}
\bar{\mathcal C}^r(\rho)=\mathcal Q^r_B(\rho)=S(\rho_B)-S(\rho).\label{th2}
\end{equation}
for maximally correlated states $\rho^{mc}=\sum_{ij}\rho_{ij}|ii\rangle\langle jj|$.

\begin{proof} Any maximally correlated state can be written in a pure state decomposition form $\rho=\sum_iq_i|\Psi_i\rangle\langle\Psi_i|$ with $|\Psi_i\rangle=\sum_j\lambda_{ij}|jj\rangle$ and $\langle\Psi_i|\Psi_{i'}\rangle=\delta_{ii'}$. Here $\rho_B=\sum_j(\sum_iq_i|\lambda_{ij}|^2)|j\rangle\langle j|$ has eigenbasis $\mathbb E_B=\{|i\rangle\}$. In order to calculate the $B$-side MID, we consider Bob's projective measurement $\Lambda_B^{\mathbb E_B}$, which takes the bipartite state to $\rho^{\mathbb E_B}=\sum_j(\sum_iq_i|\lambda_{ij}|^2)|jj\rangle\langle jj|$. Apparently, $S(\rho^{\mathbb E_B})=S(\rho_B)$. By definition, we have
\begin{equation}
\mathcal Q_B^r(\rho)=S(\rho\| \rho^{\mathbb E_B})=S(\rho^{\mathbb E_B})-S(\rho)=S(\rho_B)-S(\rho).
\end{equation}

In order to extract the maximum average coherence on Bob's side, Alice measures her quantum system on basis $\Xi_A$, where $|\xi_k^A\rangle=\frac{1}{\sqrt{d}}\sum_{j=0}^{d_A-1}e^{-\imath\frac{2\pi kj}{d_A}}|j\rangle$, $k=0,\cdots,d_A-1$ and $d_A$ is the dimension of $A$. On the measurement result $k$, Bob's state is steered to $\rho_B^{\xi_k}=\sum_iq_i|\phi^{\xi_k}_i\rangle\langle\phi^{\xi_k}_i|$ where $|\phi^{\xi_k}_i\rangle=\sum_je^{\imath\frac{2\pi kj}{d_A}}\lambda_{ij}|j\rangle$, which happens with probability $p^{\xi_k}=\frac{1}{d}$. Apparently, $\langle\phi^{\xi_k}_i|\phi^{\xi_k}_{i'}\rangle=\delta_{ii'}$ and hence $S(\rho_B^{\xi_k})=S(\rho)=-\sum_iq_i\log_2q_i$. Meanwhile, we have $\Lambda^{\mathbb E_B}_B(\rho_B^{\xi_k})=\sum_j(\sum_iq_i|\lambda_{ij}|^2)|j\rangle\langle j|=\rho_B$. The coherence of steered state $\rho_B^{\xi_k}$ is then
\begin{equation}
C^r(\rho_B^{\xi_k},\mathbb E_B)=S(\Lambda^{\mathbb E_B}_B(\rho_B^{\xi_k}))-S(\rho_B^{\xi_k})=S(\rho_B)-S(\rho),
\end{equation}
for any outcome $k$. Therefore we arrive at Eq. (\ref{th2}).
\end{proof}
\end{theorem}

Any pure bipartite state can be written in a Schmidt decomposition form $|\Psi\rangle=\sum_j\lambda_j|jj\rangle$, and hence belongs to the set of maximally correlated states. As introduced in Ref. \cite{PhysRevLett.115.020403}, a maximally correlated states $\rho^{mc}$ is prepared from an product states $\varrho_B\otimes|0\rangle_C\langle0|$ using an incoherent unitary operator, and its entanglement $E(\rho^{mc})$ can reach the initial coherence of $\varrho_B$. Further, for maximally correlated states, one can check the equality, $E(\rho^{mc})=\mathcal Q_B(\rho^{mc})$. Therefore, $\rho^{mc}$ can be used in a scenario where coherence is precious and entanglement is not as robust as single-party coherence. Precisely, consider the situation where Alice and Bob share a maximally correlated state $\rho^{mc}_{AB}$ but they are not use it in a hurry. To store the resource for latter use, she can transfer the entanglement between them into Bob's coherence using her local selective measurement. Bob stores his coherent state as well as Alice's measurement results. When required, Bob can perfectly retrieve the entanglement by preparing a maximally correlated state using only incoherent operations.

\subsection*{Two-qubit case, relation between $l_1$-norm of SIC and trace-norm distance of $B$-side MID}
One cannot define MID based on the $l_1$-norm distance, since it does not satisfy (D6) in general. However, it can be checked that for single-qubit states $\rho_B$ and $\sigma_B$, $||\rho_B-\sigma_B||_{l_1}=D^t(\rho_B,\sigma_B)=|\boldsymbol r^\rho-\boldsymbol r^\sigma|$ \cite{PhysRevA.91.042120}, where $\boldsymbol r^\rho$ and $\boldsymbol r^\sigma$ are Bloch vectors of $\rho_B$ and $\sigma_B$ respectively. Hence the $l_1$-norm of coherence for a single-qubit state $\rho_B$ can be written as
\begin{equation}
C^{l_1}(\rho_B,\Xi)=D^t(\rho_B,\Lambda^\Xi(\rho_B)).
\end{equation}
Besides, $D^t$, which satisfies condition (D6), is proper to be used as a distance measure for MID. Therefore, when the Bob's particle is a qubit, it is meaningful to study the relation between $l_1$-norm of SIC and trace-norm distance of $B$-side MID. Now we consider a two-qubit state $\rho$, and employ $C^{l_1}$ in the definition of $\bar{\mathcal C}(\rho)$ as in Eq. (\ref{SCoh}) and prove the following theorem.

\begin{theorem}
For a two-qubit state $\rho$, we have
\begin{equation}
\bar{\mathcal C}^{l_1}(\rho)=\mathcal Q_B^t(\rho).
\end{equation}
\end{theorem}

\begin{proof}The state of a two-qubit state can be written as $\rho=\frac14\sum_{i,j=0}^3\Theta_{ij}\sigma_i^A\otimes\sigma_j^B$, where the coefficient matrix $\Theta_{ij}=\mathrm{tr}(\rho\sigma_i^A\otimes\sigma_j^B)$ can be written in the block form $\Theta=\left(\begin{array}{cc}1&\boldsymbol b^\mathrm T\\ \boldsymbol a & T\end{array}\right)$.

For non-degenerate case $b\neq0$, we choose the eigenbasis of $\rho_B$ for the basis of density matrix and hence $\boldsymbol b=(0,0,b_3)$. Further, a proper basis of qubit $A$ is chosen such that the matrix $T$ is in a triangle form with $T_{11}=T_{12}=T_{21}=0$. We calculated the explicit form of $\mathcal Q_B^t(\rho)$ and $\bar{\mathcal C}^{l_1}(\rho)$ and obtain
\begin{eqnarray}
&&\mathcal Q_B^t(\rho)=\bar{\mathcal C}^{l_1}(\rho)
=\bigg[\frac{T_{22}^2+T_{31}^2+T_{32}^2}{2}\nonumber\\
&&+\frac{\sqrt{(T_{32}^2+T_{22}^2)^2+2T_{31}^2(T_{32}^2-T_{22}^2)+T_{31}^4}}{2}\bigg]^{\frac12}.
\end{eqnarray}

For degenerate case with $b=0$, we can always chose proper local basis such that $T$ is diagonal. Here we impose $T_{11}\geq T_{22}\geq T_{33}$ without loss of generality. Direct calculations lead to
\begin{equation}
\mathcal Q_B^t(\rho)=\bar{\mathcal C}^{l_1}(\rho)=T_{22}.
\end{equation}

\end{proof}

We check that, for the state $\rho=\frac12|\Phi^+\rangle\langle\Phi^+|+\frac12|01\rangle\langle01|$, we have $\bar{\mathcal C}^{r}(\rho)<\mathcal Q_B^r(\rho)$, but according to theorem 3, $\bar{\mathcal C}^{t}(\rho)=\mathcal Q_B^t(\rho)$. It means that relative entropy of coherence and $l_1$-norm of coherence are truly different measures of coherence.

\section*{Discussion}
In this paper, we have introduced the notion of SIC which characterizes the power of Alice's selective measurement to remotely create quantum coherence on Bob's site. Quantitative connection has been built between SIC and the initially shared quantum correlation measured by $B$-side MID. We show that SIC is always less than or equal to $B$-side MID.  Our results are also generalized to a tripartite scenario where Alice can build the entanglement between Bob and Charlie in a controlled way.

Next, we discuss a potential application of SIC in secrete sharing.
Suppose Alice and Bob share a two-qubit state $|\Phi\rangle=\frac{1}{\sqrt 2}(|00\rangle+|11\rangle)$, whose SIC reaches unity. When Alice measures her state on different basis, Bob's state is steered to, e.g., $\mathbb E_B^z=\{|0\rangle,|1\rangle\}$ or $\mathbb E_B^x=\{|+\rangle,|-\rangle\}$ with $|\pm\rangle=\frac{1}{\sqrt 2}(|0\rangle\pm|1\rangle)$. The coherence of states in $\mathbb E_B^z$ reach unity on basis $\mathbb E_B^x$ and vise visa. Consequently, when we measure the states in the set $\mathbb E_B^z$ on basis $\mathbb E_B^x$, the outcome is completely random. It is essential to quantum secret sharing using $|\Phi\rangle$. In this sense, the SIC is potentially related to the ability for Alice to share secret with Bob.

Coherence and various quantum correlations, such as entanglement and discord-like correlations, are generally considered as resources
in the framework of resource theories \cite{Phys.Rev.Lett.115.070503,arXiv:1506.07975}.
By coining the concept of SIC, we present an operational interpretation between
measures of those two resources, SIC and MID, and open the avenue to study
their (ir)reversibility.
The applications of various coherence quantities like SIC in many-body systems,
as in the case of entanglement \cite{CuiNC,CuiPRB,CuiPRX}, can be expected.

\bibliography{sample}

\section*{Acknowledgements}

This work was supported by NSFC under Grant Nos. 11447161, 11504205 and 91536108, the Fundamental Research Funds of Shandong
University under Grant No. 2014TB018, the National Key Basic Research Program of China under Grant No.
2015CB921003, and Chinese Academy of Sciences Grant No. XDB01010000.

\section*{Author contributions statement}

X.H. and H.F. contributed the idea. X.H. performed the calculations. X.H. and H.F. wrote the paper. All authors reviewed the manuscript and agreed with the submission.

\section*{Additional information}

\textbf{Competing financial interests:} The authors declare that they have no competing financial interests.

\end{document}